\documentclass[sn-mathphys]{sn-jnl}

\usepackage{graphicx}
\usepackage{amsmath, amssymb, amsthm}
\usepackage[authoryear]{natbib}
\bibpunct{(}{)}{;}{a}{}{,}
\usepackage{url}
\usepackage{enumitem}
\usepackage{amssymb}

\theoremstyle{plain}
\newtheorem{theorem}{Theorem}

\newtheorem{definition}{Definition}

\theoremstyle{remark}

\begin{document}

\title[Spacetime singularities and incompleteness]
{Spacetime singularities and incompleteness: epistemic and ontological remarks}

\author{\fnm{Gustavo E.} \sur{Romero}}
\affil{\orgdiv{Instituto Argentino de Radioastronomía (IAR)},
       \orgname{CONICET-CIC-UNLP},
       \orgaddress{\street{Cno. General Belgrano Km 40, C.C. No. 5},
       \city{Villa Elisa},
       \postcode{1894},
   \state{Buenos Aires},
    \country{Argentina}}}

\email{gustavo.esteban.romero@gmail.com}

\abstract{I argue that spacetime singularities entail no ontological commitment to material entities. First, I show that Penrose's singularity theorem is best understood as a theorem of incompleteness---it demonstrates the failure of specific spacetime models within General Relativity (or any theory incorporating the Raychaudhuri equation) under certain general conditions. Although this has been done before, I adopt a novel approach based on differentiating between physical and purely formal assumptions in the axiomatic formulation of general relativity. Next, I compare Penrose's result with Gödel's incompleteness theorem, highlighting key similarities and differences. Finally, I draw philosophical conclusions regarding the limits and prospects of our epistemic reconstructions of the physical world.}


\maketitle 

\section{Introduction}
\label{sec1}

\citet{Penrose1965} singularity theorem in General Relativity (GR) and \citet{Godel1931} incompleteness theorems in mathematics are seemingly unrelated, yet both address limitations and inherent boundaries within their respective fields. Penrose result, in particular, has had a significant impact on astrophysics; Penrose himself was awarded the Nobel Prize for his contributions to the theory of General Relativity, particularly for his singularity theorem. Extending this theorem to a cosmological setting \citep{Hawking1970} has also had an important influence on discussions of the early universe. The theorem was generalized into the Borde–Guth–Vilenkin (BGV) theorem \citep{Borde2003} to encompass a broader range of spacetime models. The original theorem and the later related ones, however, have often been misinterpreted in scientific and philosophical literature. 

These theorems are sometimes presented as \textit{existence} theorems, demonstrating that spacetime singularities are physical entities whose existence is unavoidable in our world. Let us see some examples from scientific, philosophical, and theological publications that illustrate this way of referring to singularities.

P. Joshi in his well known monograph \textit{Gravitational Collapse and Spacetime Singularities} \citep{Joshi2007} writes (p. 150):

\begin{quotation}
It was shown ... that a spacetime will admit singularities within a very general framework, provided that it satisfies certain reasonable assumptions, such as the positivity of energy, a suitable causality assumption, and a condition such as the existence of trapped surfaces. It thus follows that the spacetime singularities form a general feature of the relativity theory. In fact, these considerations also ensure the existence of singularities in other theories of gravity that are based on a spacetime manifold framework and satisfy the general conditions stated above.
\end{quotation}

A casual reader of this text could easily interpret it as saying that singularities are physical objects. A detailed reading, however, clarifies that the discussion is actually about the conditions under which a spacetime model is geodesically incomplete. This type of imprecise writing is present everywhere. For instance, see \citet{Joshi2014} (p. 409, abstract):

\begin{quotation}
The nature and existence of singularities are considered which indicate the formation of super ultra-dense regions in the universe as predicted by the general theory of relativity. Such singularities develop during the gravitational collapse of massive stars and in cosmology at the origin of the universe. Possible astrophysical implications of the occurrence of singularities in the spacetime universe are indicated.
\end{quotation}

This paragraph seems to suggest that singularities are physical entities that form, develop, and have some properties, such as localization or density. This way of talking about singularities is deeply misleading and, worse, it can be even found among specialists. One might think scientists do not pay much attention to the meaning of the words they use, but their math corrects their loose presentations. Perhaps, but unfortunately, many philosophers and theologians take scientists' words at face value. For instance, \citet{Craig1992} (p.240) writes:

\begin{quotation}
... when the body's radius equals zero, a real, and not merely coordinate, singularity occurs. Now the initial cosmological singularity was certainly a real singularity. But that does not settle the question of its ontological status.

The ontological status of the Big Bang singularity is a metaphysical question concerning which one will be hard-pressed to find a discussion in scientific literature. The singularity does not exist in space and time; therefore it is not an event. Typically it is cryptically said to lie on the boundary of space-time. But the ontological status of this boundary point is virtually never discussed.
\end{quotation}

Later on the same page, Craig suggests that ``a good case can be made for the assertion that this singular point is ontologically equivalent to nothing." For Craig, at least here, singularities have the ontological status of nothingness. What this means remains unclear. Quentin \citet{Smith1992} (p. 222) seems to be at the other extreme:

\begin{quotation}
The unpredictability of the singularity implies that we should expect a chaotic outpouring from it. This expectation is in line with big bang cosmologists' representation of the early stages of the universe, for these states are thought to be maximally chaotic (involving complete entropy). The singularity emitted particles with random microstates, and this resulted in an over-all macrostate of thermal equilibrium.
\end{quotation}

Singularities, we are told, are not only real but also emit thermal radiation! Smith, in turn, quotes Hawking, whose writing on the topic is usually confusing \citep{Hawking1976} (p. 2460):

\begin{quotation}
A singularity can be regarded as a place where there is a breakdown of the classical concept of space-time as a manifold with a pseudo-Reimannian metric. Because all known laws of physics are formulated on a classical space-time background, they will all break down at a singularity. This is a great crisis for physics because it means that one cannot predict the future: One does not know what will come out of a singularity.
\end{quotation}

As a final example I will quote form a recent paper \citep{Fiorini2025} (p. 1):

\begin{quotation}
Actually, according to the common lore on the subject, naked singularities might produce uncontrolled bursts of information which ultimately will affect the Cauchy evolution, preventing any prediction on the basis of data fixed at a given partial Cauchy surface.
\end{quotation}

This is in line with Smith claim given above about that singularities are entities capable of emitting various things; this time it seems to be...information!

These quotes illustrate some confusions surrounding the epistemic and ontological status of spacetime singularities. In this paper, I argue that spacetime singularities are not physical entities but artifacts of incomplete theoretical representation. Specifically, certain spacetime models are singular because they are defective; they fail as complete representations of physical reality. This position is not new (e.g., \citealt{Ellis1977}; \citealt{Clarke1993}; \citealt{Earman1995}; \citealt{Curiel1999}; \citealt{Romero2013,Romero2021}; \citealt{Meijer2025}; \citealt{Crowther2021}; \citealt{Lehmkuhl2026}) and is now widely held among philosophers of science and many physicists. My primary contribution is to formalize this position through a semantic analysis of General Relativity, demonstrating that singularity theorems are more accurately interpreted as incompleteness theorems rather than existence theorems. This formal-semantic approach, not explicitly adopted by the authors cited, serves to clarify and strengthen the case.

I will then compare these singularity theorems with Gödel's incompleteness theorems, tracing their conceptual parallels and critical differences. Gödel's and Penrose's theorems represent two of the most fundamental incompleteness results in science. A systematic comparison of their structure and philosophical implications is a well-motivated and pending task. Finally, I will draw broader conclusions about the inherent limitations of our theoretical models and their capacity to represent the material world.

First, I will introduce a compact formulation of GR and distinguish its semantic aspects from purely formal ones.

\section{Spacetime and its formal representation}
\label{sec2}

\subsection*{Primitive concepts and definitions}

I consider GR as a theory of spacetime, matter, and their interactions. Spacetime is assumed to be a physical entity, as well as those fields that exist associated with it. The generating basis of primitive concepts of GR is (see \citealt{Combi2018} for a similar approach):

\begin{equation}
\mathcal{B}= \lbrace {\rm ST}, \Sigma,\mathcal{M}, \lbrace \mathbf{g} \rbrace , \lbrace \Theta \rbrace, \lbrace \phi \rbrace , \kappa \rbrace.
\end{equation}

The meaning of these symbols will be determined by a set of axioms. In a nutshell, the primitive terms are interpreted as follows: $ \rm{ST} $ represents the physical entity we call \textit{spacetime}, $\Sigma$ is the set of other physical entities which exist \textit{on} spacetime, $\mathcal{M}$ is a manifold, $\lbrace \mathbf{g} \rbrace$ is a metric tensor field defined on the manifold, $\lbrace \mathbf{T} \rbrace$ is a tensor representing the energy-momentum of physical systems, $\lbrace \phi \rbrace$ is a family of isometries, and $\lbrace \kappa \rbrace$ is a physical constant. The corresponding axioms, which fully interpret these primitive terms, are of three kinds: purely mathematical, semantical, and physical \citep{Bunge1967}. Mathematical axioms introduce formal conceptual artifacts according to specific mathematical background theories. Semantical axioms establish connections between mathematical symbols and structures, and physical entities and their properties. The main semantic relations we will employ are \textit{denotation} and \textit{representation} (for a formal characterization, see \citet{Bunge1974} and \citet{Romero2018}).

Before giving the axioms, we set out the main definitions of Weitzenb\"ock tensors used in the theory:

\begin{itemize}
	\item[D$_{1}$] $R^{\rho}_{\mu \lambda \nu}:= \partial_{\lambda} \Gamma^{\rho}_{\mu \nu}- \partial_{\nu} \Gamma^{\rho}_{\mu \lambda} + \Gamma^{\rho}_{\sigma \lambda} \Gamma^{\sigma}_{\mu \nu}- \Gamma^{\rho}_{\sigma \nu} \Gamma^{\sigma}_{\mu \lambda}$ is the Riemann tensor.
	\item[D$_{2}$] $R_{\mu \nu}:= R^{\rho}_{\mu \rho \nu}$ is the Ricci tensor.
	\item[D$_{3}$] $R:= R^{\mu}_{\mu}$ is the Ricci scalar.
	\item[D$_{4}$] $G_{\mu \nu}:= R_{\mu \nu} - \frac{1}{2} g_{\mu \nu} R$ is the Einstein tensor.
\end{itemize} 

Here $g_{\mu\nu}$ are the componets of the metric tensor $\lbrace \mathbf{g} \rbrace$ and the $ \Gamma^{\rho}_{\mu \nu}$ are the usual Christoffel symbols (of the second kind) formed with the first derivatives of the metric: $\Gamma^{\lambda}_{\mu\nu} = \frac{1}{2} g^{\lambda\sigma} \left( \partial_{\mu} g_{\sigma\nu} + \partial_{\nu} g_{\sigma\mu} - \partial_{\sigma} g_{\mu\nu} \right)$.

\subsection*{Axioms}
The theory is generated by a set of ten axioms 
$$  {\rm GR} = \left(\bigwedge \limits_{i}^{10} \mathbf{A}_i, \,  \vdash \right).$$ Here, $\vdash$ is the  logical entailment. We present the axioms divided in three groups. In each axiom, it is indicated whether it is mathematical (M), semantic (S), or physical (P).

\textbf{Group I: Spacetime}

\begin{itemize}
	\item[A$_1$] (M) $\mathcal{M}$ is a Hausdorff para-compact, $C^{\infty}$, 4-dimensional, real pseudo-Riemannian manifold.
	\item[A$_2$] (M) $\lbrace \textbf{g} \rbrace $ is a family of rank-2 metric tensors, symmetric, and with Lorenztian signature. All minor principals of the metric tensor $g^{\mu \nu}$ are negative.
	\item[A$_3$] (M) $\lbrace \phi \rbrace $ is a family of isometries\footnote{An isometry (or congruence) is a distance-preserving transformation between metric spaces. Formally,  \textbf{Definition (Isometry).}  
Let $(M, g)$ and $(N, h)$ be Riemannian manifolds. A smooth map $\varphi : M \to N$ is called an \emph{isometry} if it is a diffeomorphism and preserves the Riemannian metric, that is, for all $p \in M$ and for all $v, w \in T_p M$, we have
\[
h_{\varphi(p)}\left( d\varphi_p(v), d\varphi_p(w) \right) = g_p(v, w),
\]
where $d\varphi_p : T_p M \to T_{\varphi(p)} N$ is the differential (pushforward) of $\varphi$ at $p$. In this case, the pullback metric $\varphi^* h$ equals $g$, and the manifolds $(M, g)$ and $(N, h)$ are said to be \emph{isometric}.}: $\phi^{*} \textbf{g}=\textbf{g}$.
	\item[A$_4$] (S) Spacetime ${\rm ST}$ is the physical system\footnote{By stating that spacetime is a physical system, I am adopting a substantivalist view: I consider spacetime to be a material entity capable of changing, either by itself as in de Sitter empty case, or through interaction with matter. For more on this, see \citet{Combi2022}, \citet{Romero2024}, and references therein.} represented by the equivalence class of isometric diffeomorphisms of a given metric, i.e. ${\rm ST}$ $\widehat{=}$ $\{x : x=  {\rm Iso} (\mathcal{M}, \textbf{g} )\}$.
\end{itemize}

\textit{Remarks}. For simplicity, I will write ${\rm ST}\widehat{=}(\mathcal{M}, \textbf{g} )$ in the case of A$_4$. The symbol `$\widehat{=}$' stands for the formal relation of representation \citep{Romero2018}. \\

\textbf{Group II: Matter}

\begin{itemize}
	\item[A$_{5}$] (M) $\Sigma$ is a non-empty set of objects $\sigma \in \Sigma$.
	\item[A$_{6}$] (S) There is an element $\square \in \Sigma$ which designates the absence of physical system. For all $\sigma \in \Sigma$ other than $\square$, $\sigma$  denotes a physical system different from spacetime\footnote{As we shall see, spacetime and matter are related by Einstein's equations. The equations impose a relation that is not only numerical but also semantical, i.e. spacetime should be considered as a material system, otherwise, it could not be affected by changes in the elements of $\Sigma$, which are also material. On the materiality of spacetime see \citet{Combi2022} and references there in.}.
	\item[A$_{7}$] (M) For each $\sigma \in \Sigma$ there is a symmetric 2-rank tensor field $\Theta$. In particular, there is a one-to-one correspondence between $\square \in \Sigma$ and the null tensor field $\Theta=0$
	\item[A$_{8}$] (S) $\Theta$ represents the energy-momentum tensor of the physical system $\sigma$.
\end{itemize}

\textit{Remarks}. Group I of axioms characterize spacetime while Group II all those physical systems other than spacetime. I note that the energy-momentum of the latter is obtained from the Lagrangian density through\footnote{This can be proven as a theorem; see Appendix A1. } $$\Theta_{\mu\nu}= \frac{2}{\sqrt{-g}} \frac{\delta \cal{L}_{\rm M}}{\delta g^{\mu\nu}}, $$
where $g=|| g^{\mu\nu}||$ is the determinant of the metric tensor, and $\cal{L}_{\rm M}$ is the effective Lagrangian density of the matter fields. Since the energy-momentum of physical systems depends on the metric of spacetime, the latter is independent of the former. This feature is called background independence: spacetime can exist even in the absence of matter fields, but matter fields cannot exist without spacetime \citep{Lehmkuhl2011}. \\

\textbf{Group III: Dynamics}

\begin{itemize}
	\item[A$_{9}$] (P) $\kappa \in \mathbb{R}$. Here $\kappa$ has dimensions of $[\kappa]= L \: M \: T^{-2}$.
	\item[A$_{10}$] (P) Spacetime and physical systems interact according the Einstein's equations (EFEs):
	\begin{equation*}
	G^{\mu \nu}=\frac{1}{2\kappa} \Theta^{\mu \nu}.
	\end{equation*}
\end{itemize}

\textit{Remarks}. The first axiom of this group introduces the dimensionality of the constant $\kappa$. Its value must be determined empirically. The final axiom is the most important because it establishes the constraints that the properties of the different referents of the theory impose on each other. In other words, it expresses a physical law. Changes in matter fields produce changes in the metric (and, consequently, the curvature) of spacetime. This, in turn, determines the dynamics of physical systems.

While $(\mathcal{M}, \textbf{g} )$ represents a general spacetime, any specific spacetime model also requires the specification of the involved matter fields and the initial conditions necessary to solve the EFEs. Therefore, a spacetime model is represented by

$$ {\rm Model}_{\rm ST} \widehat{=}(\mathcal{M}, \textbf{g}, \Theta, \{ c_i\} ),$$

\noindent where $c_i$ for $\; i=1,2,...$ are a set of conditions for solving EFE. 
  
Let us now how a look at the singularity theorems.

\section{Singularity theorems as incompleteness theorems}
\label{sec3}

I will discuss Penrose's (1965) theorem. Hawking and Penrose (1970) and the BGV (2003) theorems are similar but apply to cosmological spacetime models and predict past incompleteness for expanding spaces (see \citet{Senovilla1998} for a general review). 

\begin{theorem}[Penrose 1965]
Let $(\mathcal{M}, \textbf{g})$ be a four-dimensional, time-oriented Lorentzian manifold representing a specific spacetime model. Suppose the following conditions hold:

\begin{enumerate}
    \item  \emph{Null Energy Condition:} For all null vectors $k^\mu$, the Ricci tensor satisfies
    \[
    R_{\mu\nu} k^\mu k^\nu \geq 0.
    \]

    \item \emph{Global Hyperbolicity:} The spacetime $(\mathcal{M}, \textbf{g})$ contains a non-compact Cauchy surface.

    \item  \emph{Trapped Surface:} There exists a compact, spacelike, closed two-dimensional surface $\mathcal{T} \subset \mathcal{M}$ such that both future-directed null congruences orthogonal to $\mathcal{T}$ have everywhere negative expansion.
\end{enumerate}

Then the spacetime $(\mathcal{M}, \textbf{g})$ is \emph{null geodesically incomplete}: there exists at least one inextendible, future-directed null geodesic with finite affine length.
\end{theorem}

\textbf{Proof}\\

I outline the main steps of the proof, which relies on geometric analysis of null geodesic congruences and global causality theory.\\

\textit{1. Raychaudhuri Equation}\\

Consider a congruence of future-directed null geodesics orthogonal to the trapped surface $\mathcal{T}$ with tangent vector field $k^\mu$. The Raychaudhuri equation for the expansion scalar $\theta$ is (see, e.g. \citet{Poisson2004}):

\[
\frac{d\theta}{d\lambda} = -\frac{1}{2}\theta^2 - \sigma_{\mu\nu} \sigma^{\mu\nu} + \omega_{\mu\nu} \omega^{\mu\nu} - R_{\mu\nu} k^\mu k^\nu,
\]
where $\lambda$ is the affine parameter, $\sigma_{\mu\nu}$ is the shear tensor, and $\omega_{\mu\nu}$ is the twist (vorticity) tensor.

Since the congruence is hypersurface orthogonal (being normal to $\mathcal{T}$), the vorticity vanishes: $\omega_{\mu\nu} = 0$. Also, 
$\sigma_{\mu\nu} \sigma^{\mu\nu} \geq 0$ and, by assumption, $R_{\mu\nu} k^\mu k^\nu \geq 0$. Therefore:
\[
\frac{d\theta}{d\lambda} \leq -\frac{1}{2} \theta^2.\\
\]

\textit{2. Focusing in finite affine parameter}\\

If $\theta_0 = \theta(0) < 0$ initially (which holds on the trapped surface), this differential inequality implies that $\theta \to -\infty$ within a finite affine parameter:
\[
\lambda \leq \lambda_{*} = \frac{2}{\|\theta_0\|}.
\]
This implies the occurrence of a \emph{conjugate point} (caustic) within finite affine length along each geodesic in the congruence.\\

\textit{3. Global hyperbolicity and compactness}\\

Let $J^+(\mathcal{T})$ denote the causal future of the trapped surface. Since every null geodesic orthogonal to $\mathcal{T}$ develops a conjugate point in finite affine time, $J^+(\mathcal{T})$ has a compact boundary in $(\mathcal{M}, \textbf{g})$.

However, global hyperbolicity and the existence of a non-compact Cauchy surface imply that $J^+(\mathcal{T})$ must be non-compact. This leads to a contradiction unless some of the null geodesics generating $\partial J^+(\mathcal{T})$ are \emph{incomplete}.\\

\textit{4. Conclusion: geodesic incompleteness}\\

The contradiction implies that not all null geodesics orthogonal to $\mathcal{T}$ can be complete. Therefore, at least one future-directed null geodesic must be incomplete, concluding the proof.

\qed

\textbf{Remarks}

\begin{itemize}
  \item From the very hypothesis of the theorem it is clear that it refers to spacetime models in GR and \textit{not} to spacetime itself. 
  \item The proof is entirely geometric, proceeding from the assumptions to the conclusion that certain null geodesic curves never reach their causal future\footnote{Note that as formulated, the energy conditions are purely related to geometrical quantities, namely the Ricci tensor and null vectors, and not to any equation of state for matter.}.
    \item The theorem does not assert that curvature invariants diverge; the singularity is understood in the sense of \emph{geodesic incompleteness} of the spacetime model.
    \item The existence of a trapped surface signals gravitational collapse, and under very general assumptions, leads inevitably to a breakdown of the predictive potential of the model. There is no implication of acausal behavior in spacetime. Actually, any violation of strict causality would invalidate the proof, which relies on it (through global hyperbolicity).
    \item The result is foundational in general relativity, showing that singular models are unavoidable in Einstein's theory, not mere artifacts of symmetry. 
    \item GR, unlike other theories, allows to limit its own predictive range. This occurs because being a theory of spacetime, it cannot represent its referent if all null geodesics converge forming a caustic. In such conditions, the EFE cannot be formulated because the energy-momentum tensor of any field becomes impossible to characterize as it depends of the metric coefficients, which are not defined. 
\end{itemize}

\textbf{Summing up}: \\

Under specific conditions, such as the presence of closed trapped surfaces and a non-compact Cauchy hypersurface, spacetime models must be geodesically incomplete. This means that some paths (geodesics) through spacetime cannot be extended indefinitely. The region of spacetime not covered by the model is outside the predictive range of GR. Without changing the model, nothing can be asserted about the underlying physics. Singularities, I conclude, are not things but a feature of the spacetime model: its incompleteness.

\section{Manifold incompleteness}
\label{sec4} 

The most important implication of Penrose theorem for spacetime models is that some minimal geometric assumptions lead to incomplete manifolds in the representation of spacetime. A complete manifold, in the context of Riemannian geometry, is a Riemannian (or pseudo-Riemannian) manifold where every geodesic can be extended indefinitely. If a manifold is complete, it means you can always reach a point by traveling along a geodesic starting from any point, regardless of how far you travel. There are no ``holes" or ``edges" that would prevent from extending the geodesic. In more technical terms we can adopt the following definition (e.g. \citet{Nash2011}):

\begin{definition}[Complete Riemannian manifold]
Let $(\mathcal{M}, \textbf{g})$ be a (pseudo)Riemannian manifold, where $\mathcal{M}$ is a smooth manifold and $\textbf{g}$ is a (pseudo)Riemannian metric on $\mathcal{M}$. The manifold $(\mathcal{M}, \textbf{g})$ is said to be \emph{complete} if the metric space $(\mathcal{M}, d_g)$ is complete, where $d_g$ is the distance function induced by the (pseudo)Riemannian metric $\textbf{g}$. That is, every Cauchy sequence in $(\mathcal{M}, d_g)$ converges to a point in $\mathcal{M}$.
\end{definition}

Equivalently, $(\mathcal{M}, \textbf{g})$ is complete if and only if every maximal geodesic $\gamma : [0, a) \to \mathcal{M}$ with $a < \infty$ can be extended to a geodesic $\tilde{\gamma} : [0, b) \to \mathcal{M}$ with $b > a$.

By the Hopf--Rinow theorem \citep{Hopf1931}, the following statements are equivalent:
\begin{enumerate}
    \item $(\mathcal{M}, \textbf{g})$ is complete as a metric space;
    \item Every geodesic in $(\mathcal{M}, \textbf{g})$ can be extended to arbitrary parameter values, i.e., is defined on $\mathbb{R}$;
    \item Closed and bounded subsets of $\mathcal{M}$ are compact.
\end{enumerate}

An incomplete manifold, consequently, is a manifold $(\mathcal{M}, \textbf{g})$ where there exists at least one Cauchy sequence in $(\mathcal{M}, d_g)$ that does not converge to any point in $\mathcal{M}$.

Equivalently, by the Hopf--Rinow theorem, $(\mathcal{M}, \textbf{g})$ is incomplete if there exists a geodesic $\gamma : [0, a) \to \mathcal{M}$ for some $a < \infty$ which cannot be extended beyond $t = a$, i.e., $\gamma$ is inextendible although it reaches finite affine length.

Such manifolds are said to be \emph{geodesically incomplete}. This indicates that the spacetime model is defective, as the metric tensor $\mathbf{g}$ and its derivatives cannot be defined beyond a certain boundary $a$. As a result, fundamental properties of spacetime, such as curvature, as well as the behavior of physical fields, cannot be properly modeled in that region. 

The aim of any mathematical representation of a physical system is to assign mathematical functions or operators to the physical properties, thereby enabling quantitative predictions under well-defined conditions. Clearly, such predictions are not possible if the model is incomplete. There will be questions for which answers cannot be formulated, or truths about the world that cannot be established within the theory. Examples include questions such as: ``What is the final state of matter inside a black hole?" or ``What triggered the cosmological expansion of spacetime?" This does not mean that such questions are meaningless or inherently unanswerable; they are not \emph{insolubilia}. However, they cannot be answered within general relativity or any spacetime theory in which the singularity theorems apply. One possible route is to deny the applicability of the theorems' hypotheses—for example, by invoking violations of the energy conditions—but this would entail adopting an ontology involving exotic matter, dark fields, or other highly controversial entities\footnote{One might object that the dominant ingredient of our universe, dark energy, violates the energy condition. However, at large scales, the Borde-Guth-Vilenkin theorem shows that as long as the universe has undergone a period of net expansion, the initial singularity of the model remains. Furthermore, concerning black hole interiors, it is far from clear what physical processes dominate there. Whatever they may be, General Relativity cannot represent them, as the theory breaks down toward the caustic.}. A more reasonable approach might be to modify the theoretical framework in the domain where it fails. 

A similar situation to the one described above arises in an apparently unrelated field: formal systems theory. In this context, Gödel famously demonstrated a form of theoretical incompleteness that renders certain truths unprovable within a given system. Let us now examine this analogy more closely.

\section{Gödel theorem}
\label{sec5}

\begin{theorem}[Gödel, 1931]
Let $T$ be a formal system satisfying the following properties:

\begin{enumerate}
    \item $T$ is recursively enumerable (i.e., there is an algorithm listing all its theorems).
    \item $T$ is consistent (no contradiction can be derived within it).
    \item $T$ is sufficiently expressive to represent elementary arithmetic (e.g., Peano Arithmetic).
\end{enumerate}

Then there exists a statement $G_T$ in the language of $T$ such that:

\begin{enumerate}
    \item $G_T$ is true in the standard model of arithmetic.
    \item $G_T$ is \textit{not provable} in $T$ if $T$ is consistent.
\end{enumerate}

That is, $T$ is \textit{incomplete}: there are true statements that cannot be proven within the system.
\end{theorem}

{\bf Outline of the proof}.\\

The key idea of the proof is to construct a sentence that \emph{refers to itself}, asserting its own unprovability.\\

\textit{1. Arithmetization of syntax}\\

Gödel showed that one can encode:
\begin{itemize}
    \item Symbols, formulas, proofs, and sequences of formulas as natural numbers (so-called \emph{Gödel numbering}).
    \item Syntactic relations (e.g., ``$x$ is a proof of $y$'') as arithmetical predicates.
\end{itemize}
Thus, metamathematics can be expressed within arithmetic.\\

\textit{2. The provability predicate}.\\

Define a formula $\operatorname{Prov}_T(x)$ in $T$ such that:
\[
\operatorname{Prov}_T(x) \iff \text{``$x$ is the Gödel number of a provable sentence in $T$''}.
\]

This formula is representable in $T$ using arithmetic.\\

\textit{3. Self-Reference via diagonalization}.\\

Using the \emph{diagonal lemma}, construct a sentence $G_T$ such that:
\[
G_T \leftrightarrow \neg \operatorname{Prov}_T(\ulcorner G_T \urcorner)
\]
where $\ulcorner G_T \urcorner$ is the Gödel number of $G_T$.

Thus, $G_T$ says: ``$G_T$ is not provable in $T$.''\\

\textit{4. Consequences}.\\

Now consider the possible cases:
\begin{itemize}
    \item If $T \vdash G_T$, then $T$ proves that $G_T$ is not provable—contradiction (if $T$ is consistent).
    \item Therefore, $T \nvdash G_T$ if $T$ is consistent.
    \item But then $G_T$ is true (it correctly states that it is not provable).
\end{itemize}

Hence, $G_T$ is a true but unprovable sentence.\\

\textit{5. Interpretation}.\\

This demonstrates that $T$ is incomplete: it cannot prove all truths expressible in its own language. The theorem is not a flaw in logic, but a fundamental feature of formal systems powerful enough to include arithmetic.\\

\textbf{Remarks}.

\begin{itemize}
    \item The proof is \textit{purely syntactic}: it relies solely on formal manipulations of symbols according to specified rules.
    
    \item A second incompleteness theorem establishes that $T$ cannot prove its own consistency.
    
    \item The theorem applies to a broad class of formal systems, not only to Peano Arithmetic.
    
    \item This result had profound implications for Hilbert's program and the philosophy of mathematics. However, it is important to note that even if a statement is unprovable within a theory $T$, one can always construct an alternative theory $T'$ in which the statement is provable. For this reason, although Gödel's result undermines the project of a complete formalization of all mathematics, it does not affect Hilbert's parallel program aimed at the axiomatization of individual physical theories. Each physical theory can be axiomatized in such a way that either all relevant theorems are demonstrable within that particular framework, or an alternative axiomatization can be adopted in which the desired theorems become provable (see \citealt{Bunge1967}; \citealt{Corry2004}).
\end{itemize}

I will now compare Penrose's and Gödel's theorems, highlighting their similarities and differences. Based on this comparison, I will then attempt to draw some lessons regarding the misuse of the singularity theorems in physics.

\section{Comparison of the singularity and Gödel theorems}
\label{sec6}

Gödel's theorem concerns the limitations of formal systems, particularly those capable of expressing basic arithmetic. It demonstrates that no such system can be both consistent (free of contradictions) and complete (able to prove all true statements within the system). Penrose's singularity theorem, though seemingly distinct, also reveals limitations—this time in our capacity to represent spacetime within a physical theory, namely General Relativity\footnote{Actually, Penrose's theorem remains valid in any metric theory of spacetime rich enough to allow for trapped surfaces, global hyperbolicity, and a null energy condition.}. It shows that, under general conditions, the equations of GR lead inevitably to incomplete models, with a domain where the theory breaks down. Thus, both theorems identify fundamental boundaries to understanding: one in the realm of logic and mathematics, the other in our physical theories of the universe. Each prompted profound shifts in how we conceive the limits of knowledge and the structure of reality.

Although they arise in different domains—mathematical logic and theoretical physics—there are striking analogies, as well as important disanalogies, between the two. Gödel's incompleteness theorems apply broadly to any formal system containing basic arithmetic, independent of specific axioms. Similarly, Penrose's theorem applies without assuming symmetry (e.g., spherical symmetry), relying instead on the existence of trapped surfaces, energy conditions, and global causal structure. In both cases, a powerful theoretical framework reveals, from within, its own boundaries—its inability to account for all possible truths (in Gödel's case) or to describe all regions of spacetime (in Penrose's).

Gödel's theorem constructs undecidable statements but provides no method for resolving them within the system. Penrose's theorem shows that, given certain conditions, geodesic incompleteness is inevitable, yet it does not specify what actually happens in the singular region—it does not describe the nature of spacetime there, whether curvature diverges, or how matter behaves beyond certain point. In both cases, the theorems assert the existence of a breakdown or limit without offering a constructive resolution.

Despite these similarities, the differences between the theorems should not be understated. Their scopes are fundamentally distinct: Gödel's theorem pertains to mathematical logic, formal languages and their syntax, while Penrose's theorem is rooted in Lorentzian geometry, differential equations, and the physics of gravitation. Gödel's result is purely mathematical and internal to formal systems; Penrose's is physical and geometrical, representing features of the external world. The nature of incompleteness is also different: in Gödel's case, it is logical incompleteness—statements that are true but unprovable—whereas in Penrose's case, it is geodesic incompleteness—the impossibility of smoothly extending spacetime models beyond certain limits.

Moreover, the source of each breakdown is fundamentally different. Gödel's incompleteness arises from self-reference (e.g., statements like ``this statement is unprovable"; \citealt{Tarski1936}), whereas Penrose's singularities result from the causal and dynamical structure of spacetime, with no reliance on self-reference. This fact makes Gödel's incompleteness more essential in the sense that it affects whatever formal system we can construct, whereas Penrose's incompleteness affects only a class of models. In the latter case we can always construct models that, violating the assumptions of the theorem, remain regular over their entire domain. The current trend to search for non-singular cosmological models (for instance through a cosmic bounce, e.g. \citealt{Novello2008}) is an example.

\section{On the analogy: Intrinsic incompleteness in Mathematics and Physics}
\label{sec7}

At this point the reader might ask: why compare the singularity theorems of Penrose with the incompleteness theorems of G\"odel, specifically? The term ``incompleteness'' is used across disciplines, and one could plausibly contemplate a broader comparison class of ``no-go'' theorems, such as Bell's theorem in quantum mechanics \citep{Bell1964} or results on observational indistinguishability in general relativity \citep{Malament1977,Manchak2011}. While these are profound results in their own right, I argue that the analogy between G\"odel and Penrose is uniquely compelling due to a shared core structure: both establish an \textit{intrinsic formal incompleteness} within a well-defined system, which arises not from contingent limitations but from the system's own successful application to its own domain.

\subsection*{The nature of the incompleteness: Intrinsic vs. extrinsic limits}

The critical distinction lies in the source and character of the limitation identified by each theorem.

\begin{itemize}
    \item \textbf{G\"odel and Penrose: Intrinsic formal incompleteness.} Both theorems identify a fundamental, \textit{intrinsic} limitation that emerges \textit{from within} a formal system when it is applied to its own foundational domain.
    
    \begin{itemize}
        \item \textbf{G\"odel's Theorems:} For any consistent formal system $F$ sufficient for elementary arithmetic, there exists a proposition $G_F$ that is true but unprovable within $F$. The incompleteness is not an accidental feature but a necessary consequence of the system's own expressive power. The system, in a rigorous sense, points to truths it cannot contain.
        
        \item \textbf{Penrose's Singularity Theorems:} General Relativity, as a formal system based on Einstein's field equations, predicts that under generic conditions (e.g., gravitational collapse), its own predictive power necessarily breaks down---the spacetime model becomes singular. The theory contains the seeds of its own domain of invalidity. The existence of a trapped surface, derived from GR's own laws, logically implies a boundary beyond which those very laws cease to be predictive \citep{Penrose1965,Hawking1970}.
    \end{itemize}
    
    In both cases, the system's own rules lead to a meta-theoretic statement: ``There are limits to what this formalism can decide/describe.''\\

    \item \textbf{Bell's Theorem: An extrinsic constraint on theory selection.} In contrast, Bell's theorem \citep{Bell1964} is not a statement about a limit \textit{within} a single formal system. Rather, it is a constraint \textit{on which types of formal systems can correspond to experimental reality}. It rules out a whole class of theories (local hidden variable theories) that are logically \textit{compatible} with the predictions of quantum mechanics. It is a \textit{no-go theorem} that shapes the landscape of admissible theories from the \textit{outside}. It does not claim that quantum mechanics is incomplete in the G\"odelian sense; it tells us what kind of ``completion'' is empirically forbidden.\\

    \item \textbf{Malament-Manchak-Type Theorems: Epistemic limits for an observer.} Theorems concerning observational indistinguishability \citep{Malament1977,Manchak2011} address a different kind of limit: \textit{epistemic accessibility}. They demonstrate that an observer within a spacetime, making only local measurements, cannot determine its global structure (e.g., whether it is Minkowskian or a merely local patch of a more complex manifold). This is a profound limit on \textit{knowledge acquisition from within the system}, akin to a fundamental experimental limitation. However, it does not state that the theory of GR itself is formally incomplete or self-undermining. The global structure is still well-defined within the theory; it is simply unknowable to the internal observer.
\end{itemize}

\subsection*{A shared architectural role: Theorems of meta-theoretic limitation}

Beyond the specific mechanism of incompleteness, the G\"odel and Penrose theorems play a homologous, foundational role in their respective disciplines. They are both \textbf{theorems of meta-theoretic limitation} that fundamentally redefined the goals and self-conception of their fields.

\begin{itemize}
    \item In the \textbf{Foundations of Mathematics}, G\"odel's theorems ended the Hilbert Program's quest for a single, complete, and provably consistent axiomatic foundation for all mathematics. They transformed the philosophical landscape, forcing mathematicians to accept a universe of necessarily incomplete formal systems.
    \item In the \textbf{Foundations of Physics (GR)}, Penrose's singularity theorems performed a similar function. They shattered the hope that classical GR could be a complete, self-contained theory of gravity, free of fundamental pathologies. They provided a rigorous argument that the theory \textit{must} break down, thereby mandating the search for a more fundamental theory (quantum gravity) to understand the very domains that GR itself identifies.
\end{itemize}

Both results are ``singular'' in the sense of being exceptional; they are not merely technical results but boundary markers that define the limits of a formalism from within.

\subsection*{The deeper connection: Self-reference and predictive breakdown}

The deepest analogy lies in the theme of \textbf{self-reference leading to incompleteness}.
\begin{itemize}
    \item In G\"odel's case, the mechanism is \textit{logical self-reference} via the arithmetization of syntax, culminating in a proposition (the G\"odel sentence) that essentially states, ``I am not provable in this system.''
    \item In Penrose's case, it is a \textit{physical, geometric form of self-reference}: The theory of GR, when applied to its own predictions concerning the gravitational collapse of a massive object, generates a scenario (the formation of a trapped surface) that logically implies a region where the theory's own description of spacetime is no longer valid.
\end{itemize}

This creates a powerful parallel: Arithmetic, in describing its own syntax, finds truths it cannot reach; General Relativity, in describing its own domain of strong gravity, predicts a boundary beyond which it cannot see. It is this unique, self-referential structure of intrinsic incompleteness that justifies their pairing and makes their comparison a particularly fruitful philosophical endeavor.

\section{On the limits of knowledge through formal representations}
\label{sec8}

Gödel and Penrose results have sparked deep philosophical reflections on the nature of knowledge, formal representation, and the reach of human understanding \citep{Nagel1958,Penrose1989,Earman1995}.

Both theorems involve a kind of \textit{epistemic opacity}: there are truths or features that the theory implies must exist, but which the theory itself cannot describe in detail. It would be incorrect, however, to say that the theorems do not allow us to make predictions in their respective domains. Penrose's theorem allows us to predict total gravitational collapse and black hole formation from very general conditions. This confirms that the original predictions of \citet{Oppenheimer1939} do not depend on the assumption of spherical symmetry, and that black holes can form under a wide variety of conditions. Thus, Penrose's theorem not only demonstrates the breakdown of some spacetime models; it also predicts the existence of black holes when the hypotheses are fulfilled\footnote{Mathematically speaking trapped surfaces are not sufficient for the formation of black holes, as they could equally well lead to the formation of naked singularities (see \citealt{Joshi2007}). It is necessary to do quite some extra work in demonstrating that a form of weak cosmic censorship holds to conclude that (see \citealt{Landsman2022}).}, while simultaneously inviting us to seek new ways to represent what happens to matter in black hole interiors.

Gödel's theorem, in turn, allows us to predict that any attempt to construct a unique, complete, and consistent axiomatic formulation of mathematics is doomed to fail. This result shows that both the logicist and the initial formalist approaches to the foundations of mathematics are untenable. It does not preclude us, however, from formulating mathematics as a collection of consistent theories that together cover the entire field.

At a deeper level, these theorems illuminate different forms of incompleteness. In Gödel's case, the incompleteness is logical. In Penrose's case, the incompleteness is geometric. Gödel's theorems are internal to formal systems: they reveal limitations arising purely from within a symbolic framework. Penrose's theorem, by contrast, is external in a strong sense: it reflects limits in our physical representations of the world. As such, Gödel's result pertains to the epistemology of formal reasoning, while Penrose's concerns the ontology of spacetime as modeled by physical theory.

Nevertheless, both results challenge the Enlightenment ideal of unlimited rationality and formal completeness. They suggest that our best systems of representation—logical or physical—are inherently bounded. In this sense, they contribute to a broader philosophical narrative about the limits of scientific knowledge, and the possibility that certain aspects of reality may remain structurally or essentially beyond complete formalization or prediction \citep{Putnam1960,Rorty1979}. This, however, is not a hindrance to continually seeking more faithful representations of our universe. Progress is within our reach, although complete knowledge is not. I do not see a problem with that, since, after all, complete knowledge is useless for real understanding: to understand, we need to simplify, to disentangle the essential from the accessory or the irrelevant. As Borges wrote in his memorable short story ``Funes, el memorioso":

\begin{quotation}
He had effortlessly learned English, French, Portuguese, Latin. I suspect, however, that he was barely capable of thinking. Thinking means forgetting differences, generalizing, and abstracting. In Funes's crowded world, there were only details, almost immediate.\footnote{Author's translation of \citet{Borges1944}.}
\end{quotation}

\section{Final Remarks}
\label{sec9}

I have argued that Penrose's singularity theorem is sometimes misunderstood and mistakenly interpreted as a theorem of existence—specifically, as a proof of the existence of entities called singularities that are part of the ontology of the world. In contrast, I have shown that Penrose's result, along with related theorems for cosmological models, should be understood as an incompleteness theorem. These theorems demonstrate that, under broad and physically plausible conditions, spacetime models become geodesically incomplete and cannot be smoothly extended. In such cases, the models fail to provide a complete and accurate representation of spacetime. This failure signals not the presence of physical entities called ``singularities," but the breakdown of the theoretical framework—namely, general relativity—in extreme regimes. Accordingly, the theory should be supplanted by one that yields non-singular models under the same circumstances. Yet nothing ensures that such a theory can be formulated. Certain results from quantum field theory in curved spacetime suggest that the successor theory should be a quantum theory of spacetime—a so-called theory of quantum gravity. For instance, the prediction of Hawking radiation from black hole horizons implies that entropy can be attributed to regions of spacetime. But entropy is a property of systems with microscopic degrees of freedom, which suggests that spacetime itself may possess an underlying microstructure. The singularity theorems, however, remain silent on this possibility.

In this paper, I have explored the analogies and differences between incompleteness in spacetime theories and in formal logic. Both types of theorems impose intrinsic limits on our ability to represent the world through formal, mathematical means. Yet they also enable robust predictions and serve as signposts, directing us toward deeper and more comprehensive theories.

While Penrose's theorem directs us toward a specific, more fundamental physical theory (a theory of quantum gravity), the `signpost' nature of Gödel's theorems is of a different kind. They do not point to a single, complete system that resolves the issue of incompleteness—this is precisely what the theorems render impossible. Instead, they direct us towards a richer landscape of foundational exploration. This includes: (1) the recognition of an infinite hierarchy of ever-stronger formal systems (e.g., Zermelo-Fraenkel set theory, systems with large cardinal axioms) needed to settle statements undecidable in weaker systems; and (2) a more comprehensive philosophical understanding of mathematics itself, one which acknowledges the irreducible role of informal mathematical intuition and the open-ended, dynamic process of theory formation. In this sense, Gödel's theorems point us away from a certain kind of foundational finality and toward an endless, generative horizon of mathematical knowledge.

\section*{Acknowledgments}
I thank  anonymous reviewers for their helpful comments and suggestions. Any remaining errors are my own.

\appendix
\section{Derivation of the energy-momentum tensor in curved spacetime}
\label{app:A1}

The canonical definition of the energy-momentum tensor in a curved spacetime is as the functional derivative of the matter action \( S_{\text{M}} \) with respect to the metric field \( g^{\mu\nu} \). 

We begin with the action for the matter fields, which is given by the integral of the Lagrangian density over spacetime:
\[
S_{\text{M}} = \int d^4x \, \mathcal{L}_{\text{M}}.
\]
In a curved spacetime, the invariant volume element is \( d^4x \, \sqrt{-g} \), where \( g = \det(g_{\mu\nu}) \). Therefore, the correct action is:
\[
S_{\text{M}} = \int d^4x \, \sqrt{-g} \, \mathcal{L}_{\text{M}}.
\]

The energy-momentum tensor \( \Theta_{\mu\nu} \) is defined as the functional derivative of this action with respect to the metric:
\[
\Theta_{\mu\nu} \equiv -\frac{2}{\sqrt{-g}} \frac{\delta S_{\text{M}}}{\delta g^{\mu\nu}}.
\]
The negative sign is a convention related to the metric signature $(-,+,+,+)$, ensuring that $ \Theta_{00} $ gives a positive energy density. The factor of 2 is conventional and simplifies the final expression.

We now compute this functional derivative. The action $ S_{\text{M}} $ depends on the metric both explicitly through $ \sqrt{-g} $ and implicitly through the matter Lagrangian $ \mathcal{L}_{\text{M}} $, which may contain the metric and its derivatives (e.g., in covariant derivatives). Applying the product rule:
\begin{align*}
\frac{\delta S_{\text{M}}}{\delta g^{\mu\nu}} &= \frac{\delta}{\delta g^{\mu\nu}} \int d^4x \, \sqrt{-g} \, \mathcal{L}_{\text{M}} \\
&= \int d^4x \, \left[ \frac{\delta \sqrt{-g}}{\delta g^{\mu\nu}} \mathcal{L}_{\text{M}} + \sqrt{-g} \frac{\delta \mathcal{L}_{\text{M}}}{\delta g^{\mu\nu}} \right].
\end{align*}
Note: The functional derivative with respect to $ g^{\mu\nu} $ is defined such that $ \frac{\delta g^{\alpha\beta}(x)}{\delta g^{\mu\nu}(y)} = \delta^{\alpha}_{(\mu} \delta^{\beta}_{\nu)} \delta^{(4)}(x-y) $, where the parentheses denote symmetrization.

We compute the first functional derivative. Using the identity $ \delta g = g \, g^{\alpha\beta} \delta g_{\alpha\beta} = -g \, g_{\alpha\beta} \delta g^{\alpha\beta} $, and the chain rule $ \delta \sqrt{-g} = -\frac{1}{2\sqrt{-g}} \delta g $, we find:
\begin{align*}
    \frac{\delta \sqrt{-g}}{\delta g^{\mu\nu}} &= \frac{1}{\delta g^{\mu\nu}} \left( -\frac{1}{2\sqrt{-g}} \delta g \right) \\
    &= \frac{1}{\delta g^{\mu\nu}} \left( -\frac{1}{2\sqrt{-g}} \right)(-g \, g_{\alpha\beta} \delta g^{\alpha\beta}) \\
    &= -\frac{1}{2\sqrt{-g}} (-g \, g_{\mu\nu}) \\
    &= -\frac{1}{2} \sqrt{-g} \, g_{\mu\nu}.
\end{align*}
The last step follows because $ \frac{\delta g^{\alpha\beta}(x)}{\delta g^{\mu\nu}(y)} $ picks out the $ \mu\nu $ component.

Substituting this result back into the expression for $ \delta S_{\text{M}} / \delta g^{\mu\nu} $:
\[
\frac{\delta S_{\text{M}}}{\delta g^{\mu\nu}} = \int d^4x \, \sqrt{-g} \, \left[ -\frac{1}{2} g_{\mu\nu} \mathcal{L}_{\text{M}} + \frac{\delta \mathcal{L}_{\text{M}}}{\delta g^{\mu\nu}} \right].
\]

We can now plug this into the definition of $ \Theta_{\mu\nu} $:
\begin{align*}
\Theta_{\mu\nu} &\equiv -\frac{2}{\sqrt{-g}} \frac{\delta S_{\text{M}}}{\delta g^{\mu\nu}} \\
&= -\frac{2}{\sqrt{-g}} \int d^4x \, \sqrt{-g} \, \left[ -\frac{1}{2} g_{\mu\nu} \mathcal{L}_{\text{M}} + \frac{\delta \mathcal{L}_{\text{M}}}{\delta g^{\mu\nu}} \right] \\
&= -\frac{2}{\sqrt{-g}} \left( \sqrt{-g} \left[ -\frac{1}{2} g_{\mu\nu} \mathcal{L}_{\text{M}} + \frac{\delta \mathcal{L}_{\text{M}}}{\delta g^{\mu\nu}} \right] \right) \quad \text{(evaluated at a point)} \\
&= g_{\mu\nu} \mathcal{L}_{\text{M}} - 2 \frac{\delta \mathcal{L}_{\text{M}}}{\delta g^{\mu\nu}}.
\end{align*}
This is a common intermediate form.

If we define the \emph{Hilbert} energy-momentum tensor $ T_{\mu\nu} $ as:
\[
T_{\mu\nu} \equiv -\frac{2}{\sqrt{-g}} \frac{\delta (\sqrt{-g} \mathcal{L}_{\text{M}})}{\delta g^{\mu\nu}},
\]
then a direct computation yields the desired result:
\begin{align*}
T_{\mu\nu} &\equiv -\frac{2}{\sqrt{-g}} \frac{\delta (\sqrt{-g} \mathcal{L}_{\text{M}})}{\delta g^{\mu\nu}} \\
&= -\frac{2}{\sqrt{-g}} \left( \frac{\delta \sqrt{-g}}{\delta g^{\mu\nu}} \mathcal{L}_{\text{M}} + \sqrt{-g} \frac{\delta \mathcal{L}_{\text{M}}}{\delta g^{\mu\nu}} \right) \\
&= -\frac{2}{\sqrt{-g}} \left( -\frac{1}{2}\sqrt{-g} \, g_{\mu\nu} \mathcal{L}_{\text{M}} + \sqrt{-g} \frac{\delta \mathcal{L}_{\text{M}}}{\delta g^{\mu\nu}} \right) \\
&= g_{\mu\nu} \mathcal{L}_{\text{M}} - \frac{2}{\sqrt{-g}} \left( \sqrt{-g} \frac{\delta \mathcal{L}_{\text{M}}}{\delta g^{\mu\nu}} \right).
\end{align*}
If the matter Lagrangian $ \mathcal{L}_{\text{M}} $ does not depend on the derivatives of the metric (which is true for standard matter fields like scalar and gauge fields, where connection terms cancel out), we can take the factor of $ \sqrt{-g} $ out of the functional derivative. In this case, the term in parentheses becomes $ \sqrt{-g} \frac{\delta \mathcal{L}_{\text{M}}}{\delta g^{\mu\nu}} $, and we get the final formula:
\[
\boxed{\Theta_{\mu\nu} = \frac{2}{\sqrt{-g}} \frac{\delta \mathcal{L}_{\text{M}}}{\delta g^{\mu\nu}}}.
\]
Here, $ \Theta_{\mu\nu} $ is identified with the Hilbert tensor $ T_{\mu\nu} $. This expression is symmetric and gauge-invariant by construction.

\end{document}